\documentclass{article}
\usepackage[utf8]{inputenc}
\usepackage{hyperref}
\usepackage{bytefield}
\usepackage{color}
\usepackage{float}
\usepackage{amsmath}
\usepackage{amssymb}
\usepackage{amsfonts}
\usepackage{amsthm}
\usepackage{svg}
\newtheorem*{defn}{Definition}
\newtheorem{propn}{Proposition}
\newtheorem{lemma}{Lemma}
\newtheorem{thm}{Theorem}
\newtheorem{cor}{Corollary}
\def\aligned{\succ\mkern-10mu\prec}
\def\tprec{\prec\mkern-8mu\prec}

\hypersetup{
  colorlinks   = true, 
  urlcolor     = black, 
  linkcolor    = black, 
  citecolor    = black  
}

\title{{\huge Simple Rigs Hold Fast}} 
\author{{\LARGE T.R.I.E.}\\Kris Coward\footnote{TODAQ}, D.R. Toliver\footnotemark[1]}
\date{}

\begin{document}

\maketitle



\section{Introduction} 

An important use of computational systems
is updating the state of an object 
while preserving some set of invariants.
That object
might be a file, 
a row in a database, 
or perhaps an entry in a distributed system.
Its invariants may place limits on its 
relationships with other objects,
and generally include maintaining 
a unique identity across updates.


The system maintaining these invariants is responsible for not \emph{equivocating}\footnote{An equivocation in an object management system is much like an inconsistency in a logic system. Has this object upheld its invariants? The question is effectively meaningless in the presence of equivocation, because the answer can be yes, and also no.},
or presenting conflicting versions of an object.
We call this absence of equivocation \emph{integrity}.


Today the system providing integrity to an object also manages its state. 
Indeed, it has long been assumed that this must be the case, and that only the system managing an object's state can provide it with integrity.




We show that this assumption is wrong.
Objects can maintain the integrity of one system while having their state fully managed by another.

~~~




We begin by examining prior work,
and then define the basic structures of rigging: twists, lines, and hitches. 
We use these to construct rigs. 

A rig is a cryptographic data structure that provides integrity-at-a-distance.
Expressing an object in this structure 
allows its source of truth to be transferred
while retaining its source of integrity.

This guarantee holds despite the integrity source being completely oblivious about the object, and can be verified locally without network calls. 
The result is that an object whose state is managed on one machine, like an untrusted laptop, can have the same integrity as a highly trusted server.

We define the rig property of holding fast, 
which is a strong form of causal connection 
between an object (i.e. a line) 
and an integrity provider (also a line).

We present guilds, which are sets of rigs, and introduce the property of support that some guilds have. 
Supportiveness is a uniqueness property: 
a rig from a supportive guild is proof that no other rig in that guild 
holds a conflicting version of that object fast to that integrity provider's line.

We introduce a very restrictive guild, $G_H$, containing only single hitch rigs, and show that $G_H$ is supportive.

We present two basic operations on rigs, splicing and lashing, and prove that rigs are closed under them. We introduce a guild transformer, $\Psi(G)$, and prove that it preserves supportiveness.

We 
define $G_\uparrow$, 
the simplest guild that is closed under the splicing and lashing operations.

Finally, we show that $G_\uparrow$ is supportive.

\section{Prior Work}


%

The data structures described in this paper 
utilize idealized one-way functions,
which provide a deterministic, fixed-size output
from a given input,
and are easy to compute,
but impossible to invert, collide, or correlate.

Although these ideal one-way functions do not exist, 
their desired properties are sufficiently closely approximated by
cryptographic hash functions \cite{hashfunctions} 
for our purposes.
We acknowledge this by referring to these one-way functions, 
and their outputs, as hashes, 
but use the ideal functions for the sake of simplifying statements 
such as saying that a structure is unique, 
rather than unique up to hash collisions.


There is a causal/chronological relation between the input and output of a hash.
Taking $h$ to be a hash function,
if $y$ is some message incorporating $h(x)$, 
and $h(y)$ is the hash of that message,
then $h(x)$ must be known before $h(y)$ is computed.
This relation is used heavily in log management \cite{hashchrono},
and forms a strict partial order \cite{lamportclocks}
(i.e. transitive, antisymmetric, and non-reflexive).
We call this relationship \emph{causal precedence}, and denote it as $x \tprec y$.

The trie data structure described by René de la Briandais \cite{trieoriginal} and Edward Fredkin \cite{triememory},
when prepared with hashes in the manner described by Ralph Merkle \cite{merkletree},
provides an elegant way to fingerprint key-value data structures,
with concise proofs of membership that uniquely associate a key and its value.

Rigging implementations employ these concise proofs for operational efficacy, 
and when we speak of \emph{tries} it is this form to which we refer.
Further, when we speak of taking the hash of a trie, or including it in another structure,
the value can be taken to be its root node
when this is sufficient to verify the proofs just described.

We do not directly consider the nature of cryptographic devices for restricting updates, 
but note that there are multiple mechanisms by which this can be accomplished,
such as signatures based on public key ciphers like RSA \cite{rsacrypt},
and one-time signatures like those designed by Leslie Lamport \cite{lamportsig}.

\section{Twists}

Rigs are made of \emph{twists}.
A twist ${\tt y}$ 
is a structure containing the following hash references:

\begin{itemize}
    \item ${\tt p}({\tt y})$, its \emph{previous} twist, arranges twists sequentially into lines;
    \item ${\tt t}({\tt y})$, its \emph{tether} twist, establishes relations between lines; and
    \item $r({\tt y})$, its \emph{rigging} trie, holds key-value pairs that complete relations between lines.
\end{itemize}

\noindent



As a concrete data structure $y$ is merely a concatenation of hashes\footnote{Implementation details can be found in ``Rigging Specifications".}.  
However the selector functions above return twists, not hashes, 
allowing statements like ${\tt x} = {\tt p}({\tt y})$ to be made. 
This is standard object/pointer ambiguity, 
which we can disambiguate when needed by saying $h({\tt x}) = h({\tt p}({\tt y}))$.

Similarly, we treat {\tt p}, {\tt t}, and $r$ as functions on the space of twists,
even though the non-injectivity of hash functions technically makes them relations rather than functions\footnote{This simplification can also be treated as an assumption that our ideal one-way function $h$
is perfectly collision resistant and therefore injective.}.






~~~

\noindent
A twist ${\tt x}$ is a \emph{predecessor} of another twist ${\tt y}$
(denoted ${\tt x} \prec {\tt y}$) if and only if:
\begin{itemize}
    \item ${\tt x} = {\tt p}({\tt y})$ (the hash $h({\tt x})$ is equal to ${\tt y}$'s previous twist hash), or
    \item $\exists {\tt z} : {\tt z} \prec {\tt y}, {\tt x} = {\tt p}({\tt z})$ (there is a twist ${\tt z}$ which is a predecessor of ${\tt y}$, and which has ${\tt x}$ as its previous).
\end{itemize}

Note that ${\tt x} \prec {\tt y}$ implies ${\tt x} \tprec {\tt y}$,
due to the inclusion of each twist's direct predecessor's hash in its own structure
and the recursive construction of the predecessor relation.

A twist ${\tt y}$ is a \emph{successor} of ${\tt x}$ (denoted ${\tt y} \succ {\tt x}$)
if and only if ${\tt x}$ is a predecessor of ${\tt y}$.

When ${\tt x} = {\tt p}({\tt y})$, we may also say that ${\tt x}$ is the \emph{direct predecessor} of ${\tt y}$,
or that ${\tt y}$ is a \emph{direct successor} of ${\tt x}$.
In reference to the previous field in each twist,
we also refer to a twist's direct predecessor as its \emph{previous}.

Further, if ${\tt x} \prec {\tt y}$ or ${\tt x} = {\tt y}$,
we denote this relation as ${\tt x} \preccurlyeq {\tt y}$ (alternately ${\tt y} \succcurlyeq {\tt x}$).

It is worth noting that the relation $\prec$ is a partial order on the set of twists.
Moreover, because each twist has only one previous,
this partial order has the structure of a tree.

~~~

We classify twists based whether their tether field is null\footnote{Mathematically we can think of this as being the empty hash. Implementations use an out-of-band signal, the null hash algorithm byte, as described in Rigging Specifications.}. 

A \emph{fast twist} is tethered to the value of its tether field. The twist is fastened, or made fast, by this tethering.

A \emph{loose twist} has a null tether, and is not fastened to anything.

To work with fast twists efficiently we introduce a number of functions that take twists and produce fast twists.

We take ${^*\tt p}({\tt x})$, the \emph{fast previous} of ${\tt x}$, to be the fast twist which most directly precedes ${\tt x}$.
$$
{^*\tt p}({\tt x}) =
\begin{cases}
    {\tt p}({\tt x}),& \text{if } {\tt p}({\tt x}) \text{ is fast}\\
    {^*\tt p}({\tt p}({\tt x}))& \text{otherwise}\\
\end{cases}
$$
It is trivial to see that  ${\tt y} = {^*\tt p}({\tt x})$ iff
${\tt y}$ is fast, and ${\tt y} \prec {\tt x}$, and ${\tt y} \prec {\tt z} \prec {\tt x}$ implies ${\tt z}$ is loose.

Similarly, we take the \emph{fast tether} ${^*\tt t}({\tt x})$ to be
the most directly preceding or equal fast twist to the tether of ${\tt x}$.
$$
{^*\tt t}({\tt x}) =
\begin{cases}
    {\tt t}({\tt x}),& \text{if } {\tt t}({\tt x}) \text{ is fast}\\
    {^*\tt p}({\tt t}({\tt x}))& \text{otherwise}\\
\end{cases}
$$

\section{Lines}


Twists link together into lines. 

\begin{defn}
A \emph{line}
is a sequence of consecutive twists $[{\tt c}_0,...,{\tt c}_n]$,
where ${\tt c}_i = {\tt p}({\tt c}_{i+1})$ for each $i$.
\end{defn}

Between twists in a line, the relation ${\tt c} \preccurlyeq {\tt d}$
is a linear order\footnote{This linearity is part of the reason consecutive twists are referred to as a ``line". Rigs provide a way to maintain that linearity in the forward direction.}.
Lines contain all their \emph{sublines}, in the obvious way.


\begin{defn}
A line {\tt A} is called an \emph{enveloping line} for 
a collection of lines $\{{\tt A}_i\}_{i=0}^n$ if {\tt A} contains every twist in each ${\tt A}_i$.
\end{defn}

\begin{defn}
The lines in a collection with an enveloping line are called \emph{aligned},
denoted $\aligned\{{\tt A}_i\}_{i=0}^n$ or ${\tt A} \aligned {\tt B}$.
\end{defn}

If a pair of twists ${\tt x}$ and ${\tt y}$ are comparable by the predecessor relation
(i.e ${\tt x} \preccurlyeq {\tt y}$ or ${\tt y} \preccurlyeq {\tt x}$),
then the single-twist lines $[{\tt x}]$ and $[{\tt y}]$ are necessarily aligned.
Consequently, we use the same notation ${\tt x} \aligned {\tt y}$ to denote when twists are aligned.

This allows us to describe the tree order structure on the predecessor relation as providing that
if ${\tt x} \preccurlyeq {\tt z}$ and ${\tt y} \preccurlyeq {\tt z}$,
then ${\tt x} \aligned {\tt y}$.

On the other hand, given for example
twists ${\tt x}$, ${\tt y}$, and ${\tt z}$
with ${\tt z} = {\tt p}({\tt x})$ 
and ${\tt z} = {\tt p}({\tt y})$,
then ${\tt z} \preccurlyeq {\tt x}$ together with ${\tt z} \preccurlyeq {\tt y}$
is clearly not sufficient to imply ${\tt x} \aligned {\tt y}$. 
When two lines share a common twist but are not aligned,
we say that they are \emph{misaligned}.

Another consequence of the concept of alignment is that if ${\tt x} \aligned {\tt y}$
and ${\tt x} \tprec {\tt y}$, then it follows (from asymmetry and non-reflexivity of $\tprec$)
that ${\tt x} \prec {\tt y}$.

~~~

We introduce a restricted notion of the length of a line, where
the line being measured begins on a fast twist, 
ends on a fast twist,
and the measurement of interest is only concerned with fast twists.

\begin{defn}
A \emph{fast line} is a line that begins and ends with fast twists.
\end{defn}

\begin{defn}
The \emph{length} of a fast line is one less than the number of fast twists it contains.
\end{defn}

Roughly speaking, fast twists are treated as fenceposts,
and the length of the line is the number of sections of fence.
Equivalently, we can say that a fast line ${\tt X} = [({^*\tt p})^n({\tt x}),...,{\tt x}]$
has a length $l({\tt X}) = n$.

Length is undefined on lines that are not fast. A line can be expanded or contracted to make it fast.





\section{Hitches}

The basic unit of rigging is called a \emph{hitch}.
The operation that a hitch performs between a pair of lines is called \emph{hitching}.

A hitch $H$ involves 5 twists
\begin{itemize}
    \item ${\tt f}(H)$, the \emph{fastener} 
    \item ${\tt l}(H)$, the \emph{lead} 
    \item ${\tt m}(H)$, the \emph{meet} 
    \item ${\tt h}(H)$, the \emph{hoist} 
    \item ${\tt n}(H)$, the \emph{post}\footnote{While the other twists are identified by their first letter, {\tt p} is already in use, so we use {\tt n} because it comes sequentially after {\tt l} and {\tt m}. The logic of this sequencing comes from the definition of footline below.}
\end{itemize}

The relations between the twists in a hitch are as follows

\begin{itemize}
    \item ${\tt f}(H) = {\tt t}({\tt l}(H))$ 
          (the fastener is the tether of the lead)
    \item ${\tt l}(H) = {^*\tt p}({\tt m}(H))$
          (the lead is the fast predecessor of the meet)
    \item ${\tt m}(H) = {^*\tt p}({\tt n}(H))$
          (the meet is the fast predecessor of the post)
    \item $r({\tt h}(H)): h({\tt l}(H)) \mapsto h({\tt m}(H))$
    (the hoist's rigging trie associates lead:meet as key:value)
    \item the hoist is \emph{the first} successor of the fastener to associate a value with the lead in its rigging trie
    \item $r({\tt n}(H)): h({\tt l}(H)) \mapsto h({\tt h}(H))$
    (the post's rigging trie associates lead:hoist as key:value)
\end{itemize}

Note that while a direct inclusion 
of $h({\tt l}(H)):h({\tt m}(H))$ 
in the hoist's rigging trie is specified here, 
the only required property of the inclusion mechanism
is that it can prove the unique value 
associated with a given key\footnote{This unique value might be null. Implementations
employ a more complex inclusion mechanism, which has additional desirable properties.}. 

~~~

We also define a \emph{half-hitch}, which contains all the twists of a hitch except the post, and all relations of a hitch except the two involving the post.

~~~

With the mechanics of the hitch described,
we can now turn to the lines that are naturally required to exist in a hitch $H$
based on the required relations
\begin{itemize}
    \item ${\tt f}(H) \prec {\tt h}(H)$
    \item ${\tt l}(H) = {^*\tt p}({\tt m}(H))$ (implying ${\tt l}(H) \prec {\tt m}(H)$)
    \item ${\tt m}(H) = {^*\tt p}({\tt n}(H))$ (implying ${\tt l}(H) \prec {\tt m}(H) \prec {\tt n}(H)$)
\end{itemize}
Thus the lines that necessarily exist in every hitch are
\begin{itemize}
    \item the \emph{topline} ${\tt T}(H) = [{\tt f}(H),...,{\tt h}(H)]$
    \item the \emph{footline} ${\tt F}(H) = [{\tt l}(H),...,{\tt m}(H)]$ 
    \item \emph{extended footline} ${\tt F}^*(H) = [{\tt l}(H),...,{\tt m}(H),...,{\tt n}(H)]$ 
\end{itemize}

\filbreak
\begin{center}
\begin{samepage}
    \includegraphics[width=0.55\columnwidth]{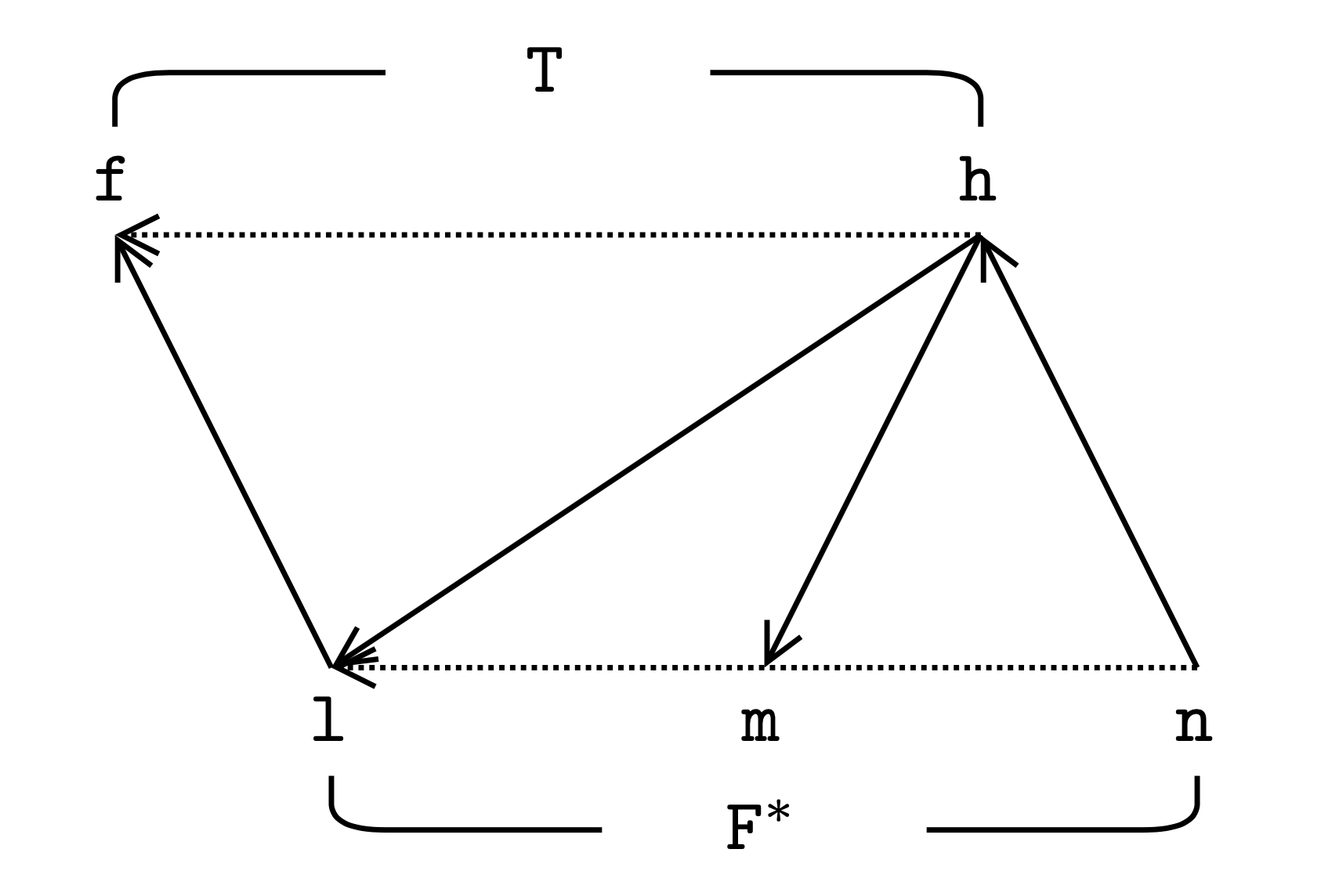}
 
 \nopagebreak   
    {\it A basic hitch. 
    Solid arrows represent direct reference: ${\tt l}$ includes the hash of ${\tt f}$.
    Dotted arrows indicate predecessors: ${\tt l} \prec {\tt m} \prec {\tt n}$.
    Time flows from left to right\footnote{Unlike an arrow.}.}
\end{samepage}
\end{center}

We say that ${\tt A}$ is \emph{hitched to} ${\tt B}$ 
iff there is a hitch $H$ in which
${\tt A} = {\tt F}(H)$ and  
${\tt B} = {\tt T}(H)$.
That is, 
${\tt A}$ is the footline of $H$ and 
${\tt B}$ is the topline of $H$.

There is a strict chronological order imposed on the twists of a hitch through hash inclusion:
\begin{equation} \label{hitchchrono}
    {\tt f}(H) \tprec {\tt l}(H) \tprec {\tt m}(H) \tprec {\tt h}(H) \tprec {\tt n}(H)
\end{equation}

\section{Rigs and Guilds}

Hitches serve as the fundamental unit of rigging, and also make up the smallest rigs.
Rigs generalize hitches and allow us to provide the same guarantees over larger structures.

First we introduce a causal relation between lines.

\begin{defn}
We say that ${\tt A}$ \emph{holds fast} to ${\tt Z}$ if
${\tt A} = [{\tt a}_\alpha,...,{\tt a}_\omega]$ and ${\tt Z} = [{\tt z}_\alpha,...,{\tt z}_\omega]$
are lines satisfying the causal relation
${\tt z}_\alpha \tprec {\tt a}_\alpha \tprec {\tt a}_\omega \tprec {\tt z}_\omega$.
\end{defn}

It is worth noting that because $\tprec$ is transitive, 
and implied by $\preccurlyeq$, 
the relation of holding fast is transitive.


\begin{center}
    \includegraphics[width=0.45\columnwidth]{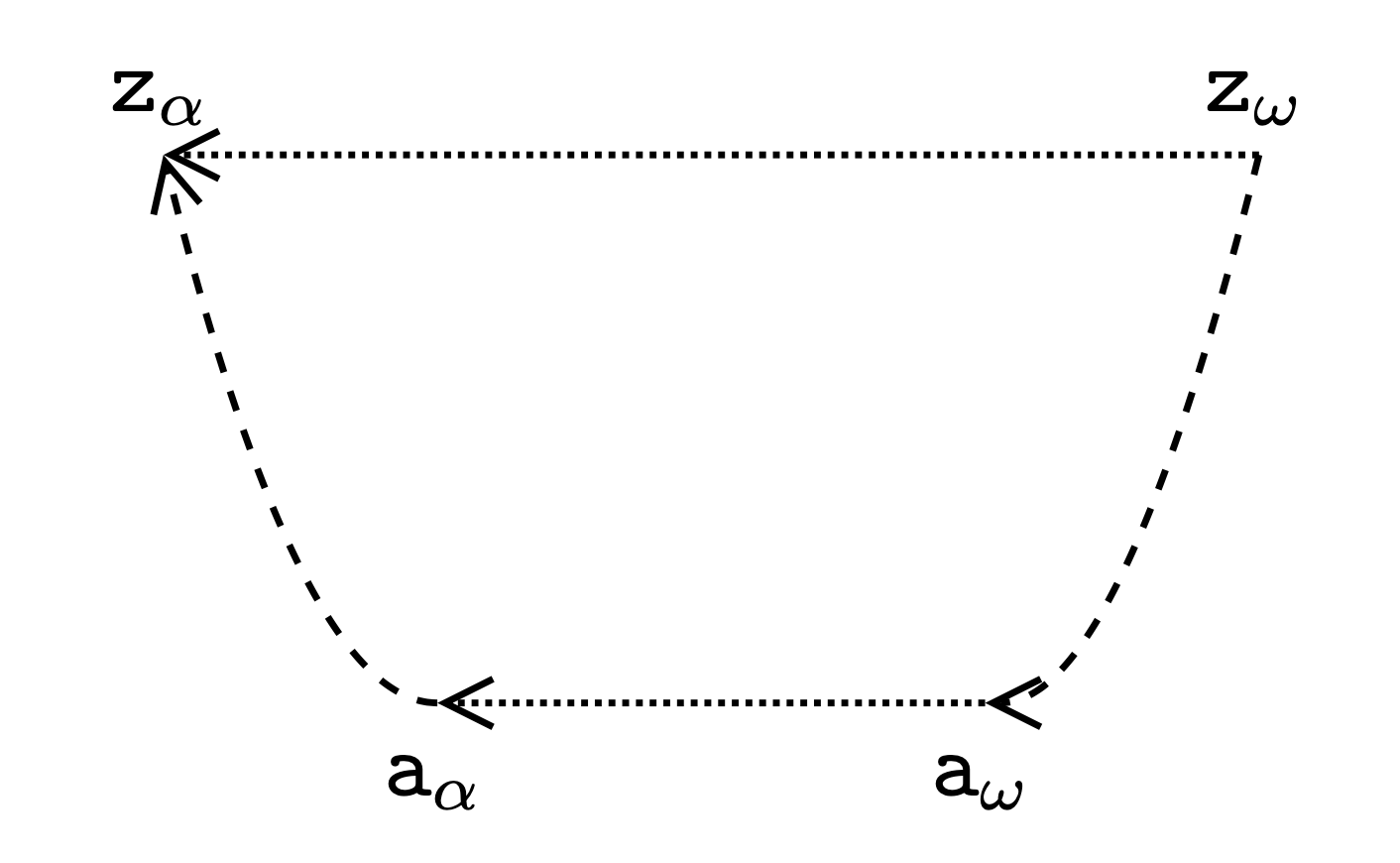}

    \nopagebreak
    {\it A holding fast to Z. 
    The dashed arrows represent chained references,
    and are curved to reinforce their ascent up the tethers
    on the past side 
    and descent down the rigging tries on the future side.}
\end{center}

\begin{defn}
A \emph{rig} $\mathfrak{R}$ is a 
collection of twists,
consisting of a \emph{corkline} ${\tt C}(\mathfrak{R})$,
a \emph{leadline} ${\tt L}(\mathfrak{R})$,
and additional twists sufficient
to demonstrate that ${\tt L}(\mathfrak{R})$ 
holds fast to ${\tt C}(\mathfrak{R})$.
\end{defn}

A rig contains the twists for all lines within it,
including its corkline and leadline\footnote{Astute readers may be wondering why this paper is called Simple Rigs Hold Fast, when in fact all rigs hold fast by definition. A linguistic shift occurred over the course of developing this work: the title came about at a time when the rigs in $G_\uparrow$ were called simple rigs, and what is now called support was referred to as holding fast. As a result the meaning of the title is a bit anticlimactic compared to ``$G_\uparrow$ is Supportive", but it does have a nice ring to it.}.

~~~

A \emph{guild} is a set of rigs.
Guilds provide constraints on acceptable rigs,
and can therefore place limits 
on the leadlines that can be rigged to a corkline.
Particularly constrained guilds have uniqueness properties which allow for useful constructions.



Our first guild, $G_H$, as described in Section \ref{G_Hsup},
is the set of rigs which consist of a single half-hitch, 
where the footline of the half-hitch is the leadline of the rig,
and the topline of the half-hitch is the corkline of the rig.

A half-hitch has the particular constraint that a given corkline 
can only support a single version of a given leadline. 
In other words, the uniqueness of successors on the corkline
guarantees the uniqueness of successors of the lead on the leadline.
This is proven below as Theorem \ref{half-hitch_sound}.

Note that the uniqueness of predecessors enables
a rig's corkline to be completely reconstructed from its last twist.
As such this last twist of the corkline
guarantees the uniqueness of the succession of twists 
from the lead forward along the leadline.

We say that a line {\tt A} is \emph{supported by} a twist ${\tt z}$
if there is a rig $\mathfrak{R}$ whose leadline is {\tt A} 
and whose corkline has ${\tt z}$ as its last twist.
In this situation, we may also say that {\tt A} is ${\tt z}$-supported, 
and that {\tt Z} supports {\tt A},
and that this support comes \emph{via} or \emph{through} $\mathfrak{R}$.

\section{Support}

We generalize this just-described uniqueness relation by elevating it to a property of guilds. The supportiveness of a guild allows its rigs to provide meaningful support.

We begin by defining a pair of relations between rigs.
The first relation is alignment, which rigs inherit from their leadlines.

\begin{defn}
Two rigs $\mathfrak{R}$ and $\mathfrak{R}'$ are \emph{aligned}
iff ${\tt L}(\mathfrak{R}) \aligned {\tt L}(\mathfrak{R}')$.
\end{defn}

\noindent
The second relation is disjointness.

\begin{defn}
Two rigs $\mathfrak{R}_0$ and $\mathfrak{R}_1$ are \emph{disjoint} if either of the following hold:
\begin{itemize}
    \item ${\tt C}(\mathfrak{R}) \not\aligned {\tt C}(\mathfrak{R}')$ (the corklines are not aligned), or
    \item ${\tt L}(\mathfrak{R}) \cap {\tt L}(\mathfrak{R}') = \varnothing$
    (the leadlines have no twists in common).
\end{itemize}
\end{defn}

Given a pair of rigs $\mathfrak{R}$ and $\mathfrak{R}'$
we adopt the following notation conventions in this paper
\begin{itemize}
    \item ${\tt L}(\mathfrak{R}) = {\tt A} = [{\tt a}_\alpha,...,{\tt a}_\omega]$
    \item ${\tt L}(\mathfrak{R}') = {\tt A}' = [{\tt a}'_\alpha,...,{\tt a}'_\omega]$
    \item ${\tt C}(\mathfrak{R}) = {\tt Z} = [{\tt z}_\alpha,...,{\tt z}_\omega]$
    \item ${\tt C}(\mathfrak{R}') = {\tt Z}' = [{\tt z}'_\alpha,...,{\tt z}'_\omega]$
\end{itemize}

Additionally, when {\tt Z} and ${\tt Z}'$ are aligned,
we take {\tt z} to denote the more recent of ${\tt z}_\omega$ and ${\tt z}'_\omega$.

Extending this notation we describe $\mathfrak{R}$ and $\mathfrak{R}'$
as being \emph{non-disjoint} iff ${\tt Z} \aligned {\tt Z}'$ and
$\exists{\tt a} : {\tt a}_\alpha \preccurlyeq {\tt a} \preccurlyeq {\tt a}_\omega
\text{ and } {\tt a}'_\alpha \preccurlyeq {\tt a} \preccurlyeq {\tt a}'_\omega$.
Given non-disjoint $\mathfrak{R}$ and $\mathfrak{R}'$ we take {\tt a} to be the common twist on their leadlines.

~~~

The relation needed can now be stated simply: when a pair of rigs are non-disjoint, they are also aligned. 
Equivalently, those rigs are either disjoint or aligned.

\begin{defn}
A pair of rigs that is neither disjoint nor aligned is \emph{misaligned}.
\end{defn}

\begin{defn}
A guild $G$ is \emph{supportive}
if it does not contain any misaligned pairs of rigs.
\end{defn}

Showing a guild is supportive provides 
a unique succession property:
if $G$ is a supportive guild,
and ${\tt Z}$ is a line, then
each twist ${\tt a}_0$ has at most one successor ${\tt a}_1$ 
such that the line $[{\tt a}_0, {\tt a}_1]$ 
is held fast to ${\tt Z}$ by rigs in $G$.

\begin{propn}
\label{unique-canonical-succession}
Given twists ${\tt a}_0$, ${\tt a}_1$, and ${\tt a}'_1$
with ${\tt p}({\tt a}_1) = {\tt p}({\tt a}'_1) = {\tt a}_0$
and non-disjoint rigs $\mathfrak{R}, \mathfrak{R}'$ in a supportive guild $G$
with $[{\tt a}_0, {\tt a}_1]$ a subline of {\tt A}
and $[{\tt a}_0, {\tt a}'_1]$ a subline of ${\tt A}'$;
then ${\tt a}_1 = {\tt a}'_1$.
\end{propn}

\begin{proof}
Because $G$ is supportive $\mathfrak{R}$ and $\mathfrak{R}'$ cannot be misaligned.
Because $\mathfrak{R}$ and $\mathfrak{R}'$ are neither disjoint nor misaligned,
they must be aligned. Therefore ${\tt A} \aligned {\tt A}'$.

Thus ${\tt a}_1$ and ${\tt a}'_1$ are both twists in the enveloping line of {\tt A} and ${\tt A}'$,
and having the same previous they must be therefore be the same twist.
\end{proof}

A supportive guild guarantees that
the leadline of a rig has the same integrity as its corkline.

\section{$G_H$ is supportive}
\label{G_Hsup}

We previously noted that a half-hitch meets the requirements of a rig: 
it has a corkline bounded by its fastener and hoist;
a leadline bounded by its lead and meet;
and it holds fast on account of relation \ref{hitchchrono},
specifically the fragment
$${\tt f}(H) \tprec {\tt l}(H) \tprec {\tt m}(H) \tprec {\tt h}(H)$$.

\begin{center}
    \includegraphics[width=0.50\columnwidth]{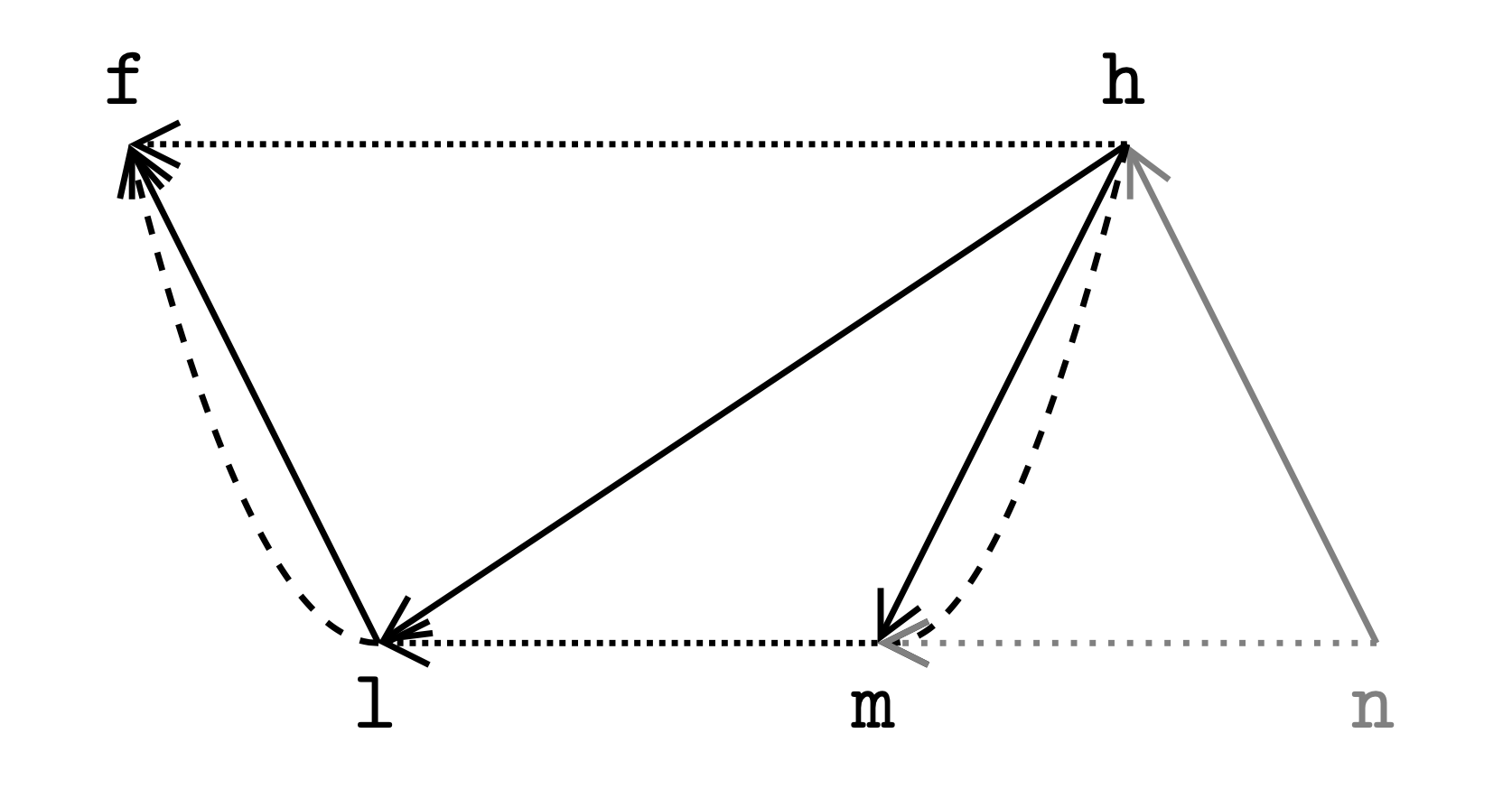}
\end{center}

\begin{defn}
The guild $G_H$ is the guild consisting of all rigs $\mathfrak{R}$ for which there is a half-hitch $H$ satisfying:
\begin{itemize}
    \item ${\tt F}(H) = {\tt L}(\mathfrak{R})$, and
    \item ${\tt T}(H)$ is a subline of ${\tt C}(\mathfrak{R})$.
\end{itemize}
\end{defn}

\noindent
In other words, $G_H$ is the guild of rigs each consisting of a single half-hitch.

\begin{thm}
\label{half-hitch_sound}
$G_H$ is supportive.
\end{thm}

\begin{proof}
Let $\mathfrak{R}$ and $\mathfrak{R}'$ be rigs in $G_H$.
Disjoint rigs can't be misaligned,
therefore we take $\mathfrak{R}$ and $\mathfrak{R}'$ to be non-disjoint
and show they are aligned.

Recall that the construction of a hitch requires
${\tt a}_\alpha$, ${\tt a}'_\alpha$, ${\tt a}_\omega$, and ${\tt a}'_\omega$
to be fast, and all the other twists in {\tt A} and ${\tt A}'$ to be loose.
Given {\tt a} as the common twist between {\tt A} and ${\tt A}'$,
we note that because ${\tt a}_\alpha \preccurlyeq {\tt a}$, 
and ${\tt a}'_\alpha \preccurlyeq {\tt a}$, 
we have ${\tt a}_\alpha \aligned {\tt a}'_\alpha$,
and can assume without loss of generality that ${\tt a}_\alpha \preccurlyeq {\tt a}'_\alpha$.

If {\tt a} is fast, then one of the following must be true:
\begin{itemize}
    \item ${\tt a} = {\tt a}_\omega = {\tt a}'_\omega$,
    in which case ${\tt a}_\alpha = {^*\tt p}({\tt a}_\omega) = {^*\tt p}({\tt a}'_\omega) = {\tt a}'_\alpha$,
    giving that ${\tt A} = {\tt A}'$, and alignment is trivial;
    \item ${\tt a} = {\tt a}_\omega = {\tt a}'_\alpha$,
    in which case $[{\tt a}_\alpha,...,{\tt a}_\omega = {\tt a}'_\alpha,...,{\tt a}'_\omega]$
    is an enveloping line for {\tt A} and ${\tt A}'$, providing alignment; or
    \item ${\tt a} = {\tt a}_\alpha = {\tt a}'_\alpha$, which is examined below.
\end{itemize}

If {\tt a} is loose, then we get ${\tt a}_\alpha = {^*\tt p}({\tt a}) = {\tt a}'_\alpha$,
so we need only demonstrate alignment when ${\tt a}_\alpha = {\tt a}'_\alpha$.

Because $\mathfrak{R}$ is in $G_H$, there is a half-hitch $H$ with
\begin{itemize}
    \item ${\tt l}(H) = {\tt a}_\alpha$
    \item ${\tt m}(H) = {\tt a}_\omega$
    \item ${\tt z}_\alpha \preccurlyeq {\tt f}(H) \prec {\tt h}(H) \preccurlyeq {\tt z}_\omega \preccurlyeq {\tt z}$
\end{itemize}

Similarly, we get $H'$ from $\mathfrak{R}'$ with
\begin{itemize}
    \item ${\tt l}(H') = {\tt a}'_\alpha$
    \item ${\tt m}(H') = {\tt a}'_\omega$
    \item ${\tt z}'_\alpha \preccurlyeq {\tt f}(H') \prec {\tt h}(H') \preccurlyeq {\tt z}'_\omega \preccurlyeq {\tt z}$
\end{itemize}

Because the fastener of a hitch is equal to the tether of the hitch's lead, we have
${\tt f}(H) = {\tt t}({\tt a}_\alpha) = {\tt t}({\tt a}'_\alpha) = {\tt f}(H')$.
Thus, the hoists ${\tt h}(H)$ and ${\tt h}(H')$
are both defined to be the earliest twist ${\tt z}_{\tt h}$
on the line $[{\tt f}(H),...,{\tt z}]$
for which the rigging trie has ${\tt a}_\alpha$ as a key.
Therefore they are the same twist, with the same rigging trie,
which associates to ${\tt a}_\alpha$
the same value ${\tt m}(H) = {\tt m}(H')$.

Consequently we have that
${\tt A} = [{\tt l}(H),...,{\tt m}(H)] = [{\tt l}(H'),...,{\tt m}(H')] = {\tt A}'$,
from which it follows trivially that ${\tt A} \aligned {\tt A}'$.
\end{proof}


\section{Splicing}

To construct rigs that span longer lines requires a composition operation
that allows a rig's leadline to be lengthened by one hitch\footnote{In order to more easily describe
the relations between the consecutive hitches,
we define our splice to perform its lengthening on the past side rather than the future side.
It may be more practical for implementations to perform the extensions on the other side,
as new hitches become available.}.  
In particular, the meet of a hitch is a fast twist, and therefore has its own tether and may be used as the lead in a consecutive hitch.


We define such an operation, called \emph{splicing},
and use it to construct rigs whose leadline {\tt L} has a length $l({\tt L}) > 1$.
To simplify descriptions of splicing we define the length of a rig to be the length of its leadline,
i.e. $l(\mathfrak{R}) = l({\tt L}(\mathfrak{R}))$. 

\begin{defn}
A pair of rigs $(\mathfrak{R}_0, \mathfrak{R}_1)$ is said to be \emph{spliceable} iff
\begin{itemize}
    \item $l(\mathfrak{R}_0) = 1$
    \item ${\tt L}(\mathfrak{R}_0) = [{\tt a}_0,...,{\tt a}_1]$
    \item there is a half-hitch $H_0$ in $\mathfrak{R}_0$ with ${\tt l}(H_0) = {\tt a}_0$ and ${\tt m}(H_0) = {\tt a}_1$
    \item $l(\mathfrak{R}_1) = n \geq 1$
    \item ${\tt L}(\mathfrak{R}_1) = [{\tt a}_1,{\tt a}_2,...,{\tt a}_{n+1}]$
    \item there is a half-hitch $H_1$ in $\mathfrak{R}_1$ with
    ${\tt l}(H_1) = {\tt a}_1$ and ${\tt m}(H_1) = {\tt a}_2$
    \item there is a hitch $H^*_0$ which adds post ${\tt p}(H^*_0) = {\tt a}_2$ to $H_0$ 
    \item ${\tt C}(\mathfrak{R}_0) \aligned {\tt C}(\mathfrak{R}_1)$
\end{itemize}

\end{defn}

\begin{defn}
The \emph{splice} of a spliceable pair of rigs $(\mathfrak{R}_0, \mathfrak{R}_1)$,
denoted $\mathfrak{R}_0 \Psi \mathfrak{R}_1$,
is the structure consisting of 
the union of each rig's collection of twists.
Its corkline is the 
minimal enveloping line of ${\tt C}(\mathfrak{R}_0)$ 
and ${\tt C}(\mathfrak{R}_1)$,
and its leadline is 
the concatenation of their leadlines.
\end{defn}

\begin{propn}
If ($\mathfrak{R}_0$, $\mathfrak{R}_1$) is a splicable pair of rigs,
then their splice $\mathfrak{S} = \mathfrak{R}_0 \Psi \mathfrak{R}_1$ is also a rig.
\end{propn}

\begin{proof}
Taking $\mathfrak{S} = \mathfrak{R}_0 \Psi \mathfrak{R}_1$,
${\tt z}_\alpha$ to be the earliest twist in ${\tt C}(\mathfrak{S})$,
and ${\tt z}_0$ to be the earliest twist in ${\tt C}(\mathfrak{R}_0)$,
it follows from $\mathfrak{R}_0$ being a rig
that ${\tt z}_\alpha \tprec {\tt z}_0 \tprec {\tt a}_0$.

Similarly taking ${\tt z}_\omega$ to be the latest twist in ${\tt C}(\mathfrak{S})$,
and ${\tt z}_{n+1}$ to be the latest twist in ${\tt C}(\mathfrak{R}_1)$,
we get ${\tt a}_{n+1} \tprec {\tt z}_{n+1} \tprec {\tt z}_\omega$.

Using the fact that ${\tt L}(\mathfrak{S})$ is a line to combine and reduce these
yields ${\tt z}_\alpha \tprec {\tt a}_0 \tprec {\tt a}_{n+1} \tprec {\tt z}_\omega$.

Thus the leadline ${\tt L}(\mathfrak{S})$ is held fast
to the corkline ${\tt C}(\mathfrak{S})$
by the splice $\mathfrak{S} = \mathfrak{R}_0 \Psi \mathfrak{R}_1$,
which is a rig.
\end{proof}

To illustrate this, below is a diagram of a pair of rigs being spliced, which also shows the hash inclusions of the half-hitch forming the bottom of $\mathfrak{R}_0$.
\begin{center}
    \includegraphics[width=0.90\columnwidth]{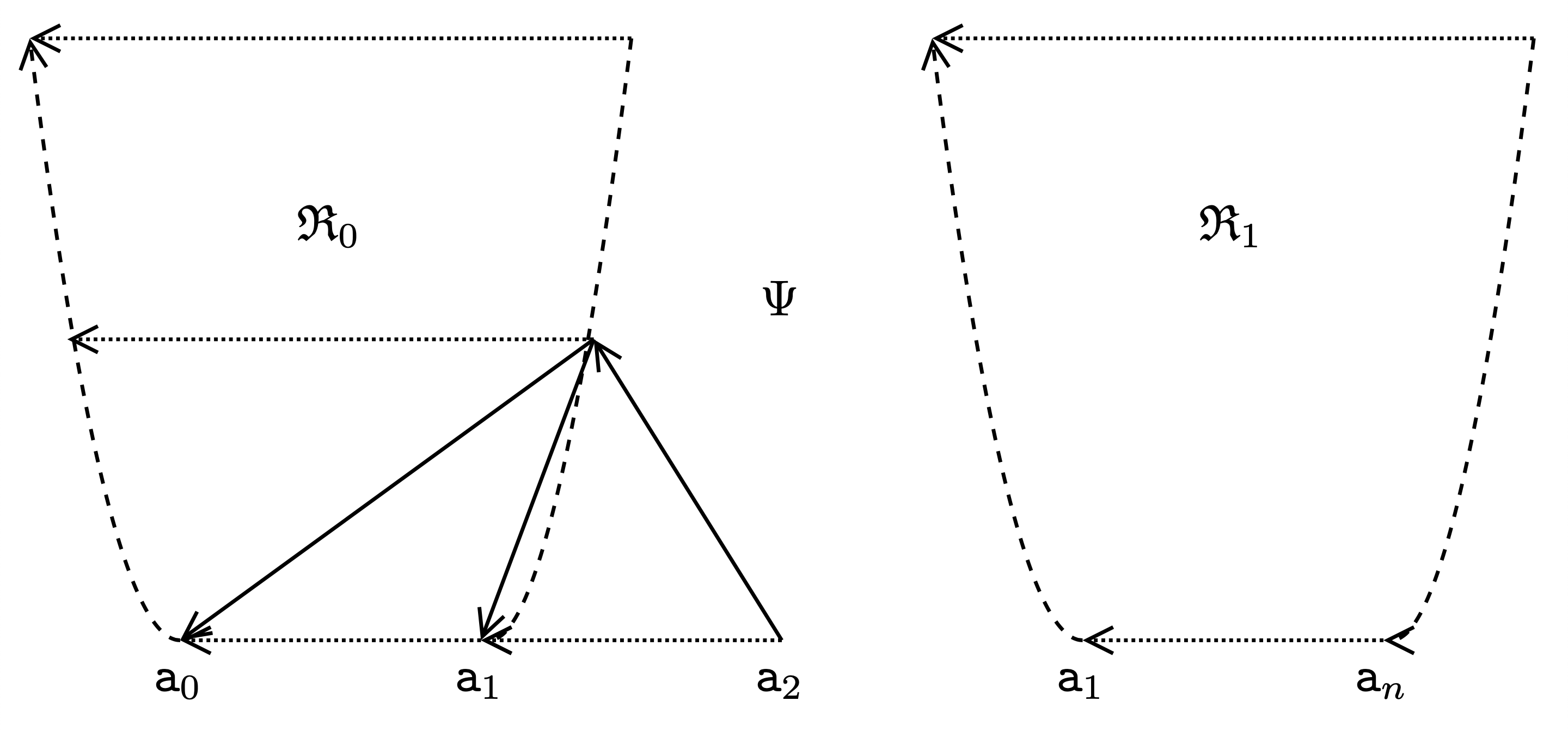}
\end{center}

And this is the rig resulting from splicing them.

\begin{center}
    \includegraphics[width=0.75\columnwidth]{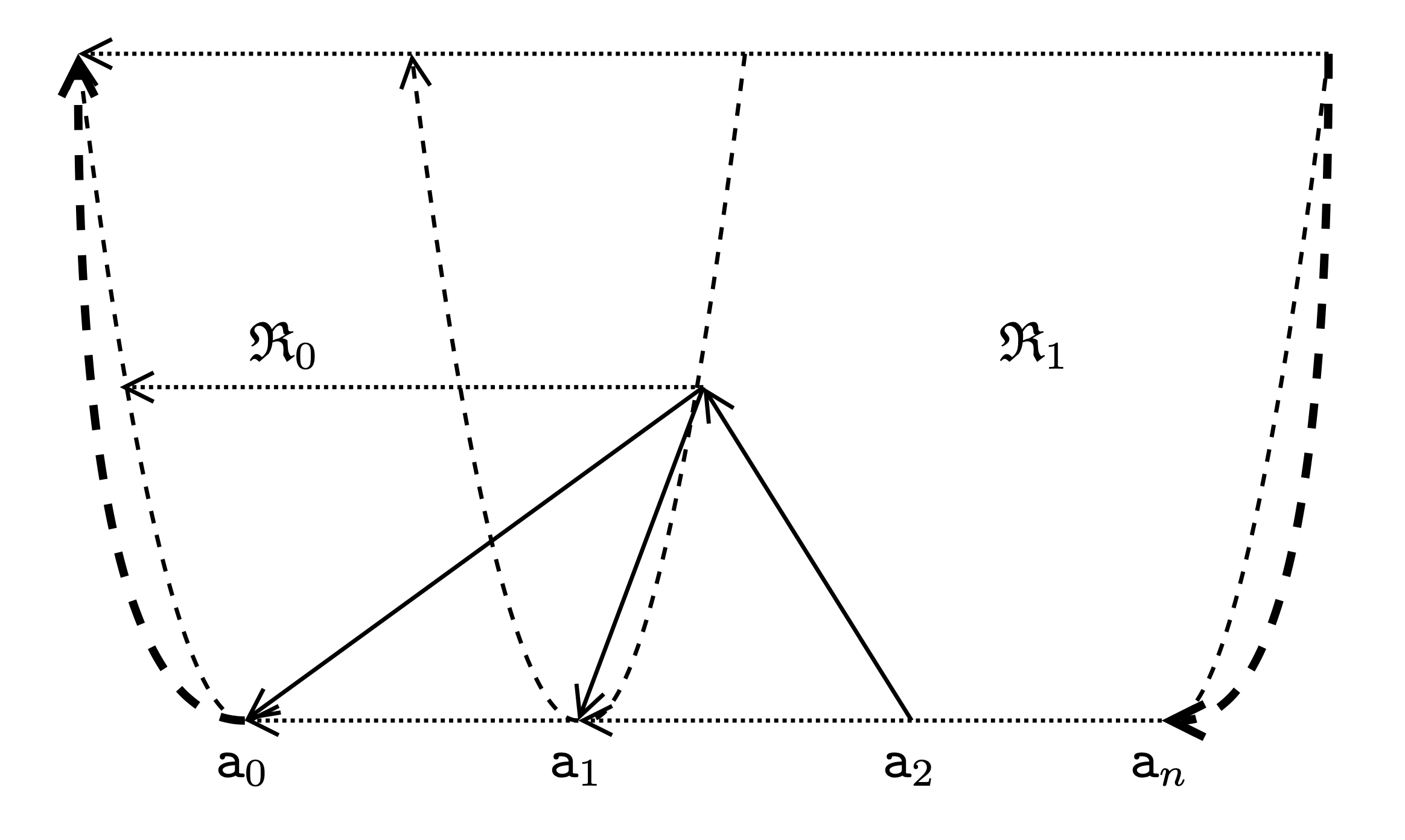}
\end{center}

As a more concrete example, here is a hitch spliced to a half-hitch.

\begin{center}
    \includegraphics[width=0.60\columnwidth]{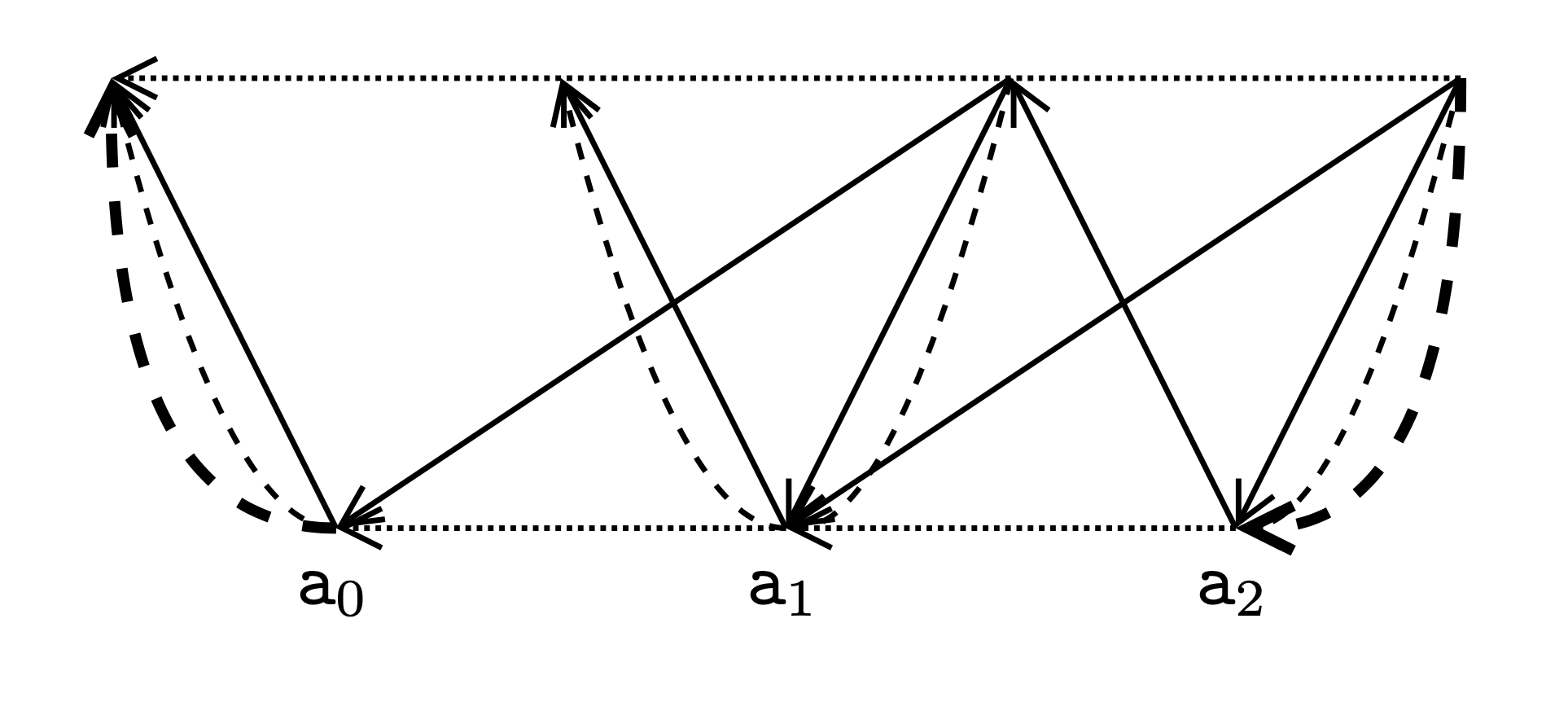}
\end{center}

\section{$\Psi(G)$ is supportive}
\label{psiG_is_supportive}

We introduce a guild transformation operator, $\Psi(G)$, which forms the closure of a guild under splicing.

\begin{defn}
Given a guild $G$ where all rigs have length 1,
let \emph{$\Psi(G)$} be the guild consisting of any rig that is either:
\begin{itemize}
    \item a rig from $G$, or
    \item obtained by splicing a rig $\mathfrak{R}_0$ from $G$ to a rig $\mathfrak{R}_1$ from $\Psi(G)$.
\end{itemize}
\end{defn}

 
To prove that this operation preserves the property of supportiveness,
we start with the more restrictive case where the leadlines begin with the same twist.

\begin{lemma}
\label{splicing-aligned}
Given a supportive guild $G$,
and rigs $\mathfrak{R}$ and $\mathfrak{R}'$ from $\Psi(G)$
with ${\tt a}_\alpha = {\tt a}'_\alpha$,
then $\mathfrak{R}$ and $\mathfrak{R}'$ are not misaligned.
\end{lemma}

\begin{proof}
Since disjoint rigs cannot be misaligned, we need only consider the non-disjoint case.

Letting $l$ and $l'$ be the lengths of $\mathfrak{R}$ and $\mathfrak{R}'$ respectively,
we assume without loss of generality that $l' \geq l$.

We prove by induction on $l$ that ${\tt A} \aligned {\tt A}'$ for all $l \geq 1$.

~~~

\textbf{Base case ($l = 1$):}

In this case $\mathfrak{R}$, being of length 1, is a rig from the supportive guild $G$,
and by definition cannot be misaligned with any other rig from $G$.

If $l' = 1$, then $\mathfrak{R}'$ is also a rig from $G$, so ${\tt A} \aligned {\tt A}'$.

If $l' > 1$, then $\mathfrak{R}'$ is obtained by splicing the rig $\mathfrak{R}'_0$
to the rig $\mathfrak{R}'_1$ at the common twist ${\tt a}'_1$ in their leadlines.
Further, $\mathfrak{R}'_0$ is a rig from $G$ 
with ${\tt A}'_0 = L(\mathfrak{R}'_0)$.

Since $\mathfrak{R}$ and $\mathfrak{R}'_0$ are both rigs in $G$
(and $G$ is supportive), they cannot be misaligned.

Further, both ${\tt A}'_0$ and {\tt A} begin with ${\tt a}_\alpha$ and have length 1,
so not only are they aligned, but in fact equal.

Then because ${\tt A}'_0$ is a subline of ${\tt A}'$
it follows that ${\tt A}' \aligned {\tt A}$.

~~~

\textbf{Induction step ($l > 1$):}

In this case, $\mathfrak{R}$ is obtained by splicing the rig $\mathfrak{R}_0 \in G$
to the rig $\mathfrak{R}_1$ at their common twist ${\tt a}_1$.
Likewise, $\mathfrak{R}'$ is obtained by splicing the rig $\mathfrak{R}'_0 \in G$
to the rig $\mathfrak{R}'_1$ at their common twist ${\tt a}'_1$.

Now, both $\mathfrak{R}_0$ and $\mathfrak{R}'_0$ are rigs from $G$,
with ${\tt A}_0 = {\tt L}(\mathfrak{R}_0)$ and
${\tt A}'_0 = {\tt L}(\mathfrak{R}'_0)$ both beginning at ${\tt a}_\alpha$.
Recalling that {\tt z} denotes the most recent of ${\tt z}_\omega$ and ${\tt z}'_\omega$,
both ${\tt C}(\mathfrak{R}_0)$ and ${\tt C}(\mathfrak{R}'_0)$
end at predecessors of {\tt z}
and are therefore aligned with each other.
Thus $\mathfrak{R}_0$ and $\mathfrak{R}'_0$ are non-disjoint.

$G$ being supportive then gives that ${\tt A}_0 \aligned {\tt A}'_0$.

Because ${\tt A}_0$ and ${\tt A}'_0$ are aligned,
both begin at ${\tt a}_\alpha$,
and both have a length of $1$,
it follows that ${\tt A}_0 = {\tt A}'_0$
and by extension that ${\tt a}_1 = {\tt a}'_1$.

Thus we have that the leadlines ${\tt A}_1 = {\tt L}(\mathfrak{R}_1)$ and
${\tt A}'_1 = {\tt L}(\mathfrak{R}'_1)$ both begin with ${\tt a}_1$.

The same argument used to show alignment of ${\tt C}(\mathfrak{R}_0)$ and ${\tt C}(\mathfrak{R}'_0)$
can also be used on ${\tt C}(\mathfrak{R}_1)$ and ${\tt C}(\mathfrak{R}'_1)$
to get that $\mathfrak{R}_1$ and $\mathfrak{R}'_1$ are non-disjoint.

So our induction hypothesis provides that $\mathfrak{R}_1$ and $\mathfrak{R}'_1$,
being non-disjoint rigs from $\Psi(G)$ with length less than $l$, must be aligned,
along with their leadlines ${\tt A}_1 \aligned {\tt A}'_1$.
Further, because {\tt A} and ${\tt A}'$ can be obtained by tracing ${\tt A}_1$ and ${\tt A}'_1$
back to ${\tt a}_\alpha$, 
it follows that ${\tt A} \aligned {\tt A}'$.
\end{proof}

This more restrictive case being proven, we can now move to the general case.

\begin{thm}
\label{splicing-sound}
Given a supportive guild $G$
where all rigs have length 1,
then $\Psi(G)$ is also supportive.
\end{thm}

\begin{proof}
Since disjoint rigs are trivially not misaligned,
we take the rigs $\mathfrak{R}$ and $\mathfrak{R}'$ in $\Psi(G)$ to be non-disjoint
and show ${\tt A} \aligned {\tt A}'$.

We assume without loss of generality that ${\tt a}_\alpha \preccurlyeq {\tt a}'_\alpha$,
and because leadlines in $\Psi(G)$ must begin with fast twists, 
we note that this can be expressed as ${\tt a}_\alpha = ({^*\tt p})^n({\tt a}'_\alpha)$
for some $n \geq 0$,
and proceed by induction on $n$.

~~~

\textbf{Base case ($n = 0$):}

${\tt a}_\alpha = ({^*\tt p})^0({\tt a}'_\alpha) = {\tt a}'_\alpha$,
so the base case is simply Lemma \ref{splicing-aligned}.

~~~

\textbf{Induction step ($n > 0$):}

If $l({\tt A}) = 1$, then fastness requires that ${\tt a}'_\alpha$
must be either ${\tt a}_\alpha$ or ${\tt a}_\omega$,
but $n > 0$ implies ${\tt a}'_\alpha \neq {\tt a}_\alpha$,
so ${\tt a}'_\alpha = {\tt a}_\omega$,
and we have an enveloping line
$$[{\tt a}_\alpha,...,{\tt a}_\omega = {\tt a}'_\alpha,...,{\tt a}'_\omega]$$
which provides ${\tt A} \aligned {\tt A}'$.

When $l({\tt A}) > 1$, then $\mathfrak{R}$ is the result of splicing a rig $\mathfrak{R}_0$ from $G$,
with ${\tt L}(\mathfrak{R}_0) = [{\tt a}_\alpha,...,{\tt a}_1]$,
onto a rig $\mathfrak{R}_1$ from $\Psi(G)$,
with ${\tt L}(\mathfrak{R}_1) = [{\tt a}_1,...,{\tt a}_\omega] = {\tt A}_1$.

Our induction precondition that $n > 0$ provides that
${\tt a}_1 = ({^*\tt p})^{n-1}({\tt a}'_\alpha) \preccurlyeq {\tt a}'_\alpha$.

Additionally, from our induction hypothesis, we have that ${\tt A}_1 \aligned {\tt A}'$;
i.e. that there is an enveloping line $[{\tt a}_1,...,{\tt a}]$,
where {\tt a} is the more recent of ${\tt a}_\omega$ and ${\tt a}'_\omega$.
Extending this enveloping line by ${\tt L}(\mathfrak{R}_0)$
to $[{\tt a}_\alpha,...,{\tt a}_1,...,{\tt a}]$, we get an enveloping line for {\tt A} and ${\tt A}'$.

Thus ${\tt A} \aligned {\tt A}'$,
and by induction we have that the $\Psi$ operation preserves supportiveness.
\end{proof}

\section{Lashings}

The splicing operation provides the ability to form rigs that provide support for arbitrarily long lines. 
The lashing operation is its compliment: it provides the ability to construct rigs with intermediate lines between the corkline and the leadline.


\begin{defn}
A pair of rigs $(\mathfrak{H}_0, \mathfrak{R}_1)$ is \emph{lashable} when it has the following properties
\begin{itemize}
    \item $\mathfrak{H}_0$ is a rig in $G_H$ (i.e. a half-hitch)
    \item $\mathfrak{R}_1$ is an arbitrary rig
    \item ${\tt C}(\mathfrak{H}_0) = [{\tt b}_0,...,{\tt b}_1]$
    \item ${\tt L}(\mathfrak{R}_1) = [{\tt b}'_0,...,{\tt b}_0,...,{\tt b}_1,...,{\tt b}'_1]$
    \item none of the twists from the corkline ${\tt C}(\mathfrak{R}_1)$ are present in $\mathfrak{H}_0$
\end{itemize}
\end{defn}

\begin{defn}
When $(\mathfrak{H}_0, \mathfrak{R}_1)$ is a lashable pair of rigs,
the lashing $\mathfrak{H}_0 \Xi \mathfrak{R}_1$
is taken to be the structure consisting of the union of
the collections of twists from $\mathfrak{H}_0$ and $\mathfrak{R}_1$,
and is assigned a corkline ${\tt C}(\mathfrak{L}) = {\tt C}(\mathfrak{R}_1)$
and a leadline ${\tt L}(\mathfrak{L}) = {\tt L}(\mathfrak{H}_0)$.
\end{defn}

\begin{propn}
If $(\mathfrak{H}_0$, $\mathfrak{R}_1)$ is a lashable pair of rigs,
then $\mathfrak{L} = \mathfrak{H}_0 \Xi \mathfrak{R}_1$ is a rig.
\end{propn}

\begin{proof}
Taking ${\tt C}(\mathfrak{R}_1) = [{\tt z}_\alpha,...,{\tt z}_\omega]$,
and ${\tt L}(\mathfrak{H}_0) = [{\tt a}_\alpha,...,{\tt a}_\omega]$,
it follows from $\mathfrak{H}_0$ and $\mathfrak{R}_1$ holding fast that
${\tt z}_\alpha \tprec {\tt m}'_0 \tprec {\tt m}_0 \tprec {\tt a}_\alpha \tprec
{\tt a}_\omega \tprec {\tt m}_1 \tprec {\tt m}'_1 \tprec {\tt z}_\omega$.

So $\mathfrak{L}$ holds ${\tt L}(\mathfrak{L})$ fast to ${\tt C}(\mathfrak{L})$,
and is therefore a rig.
\end{proof}

\begin{center}
    \includegraphics[width=0.80\columnwidth]{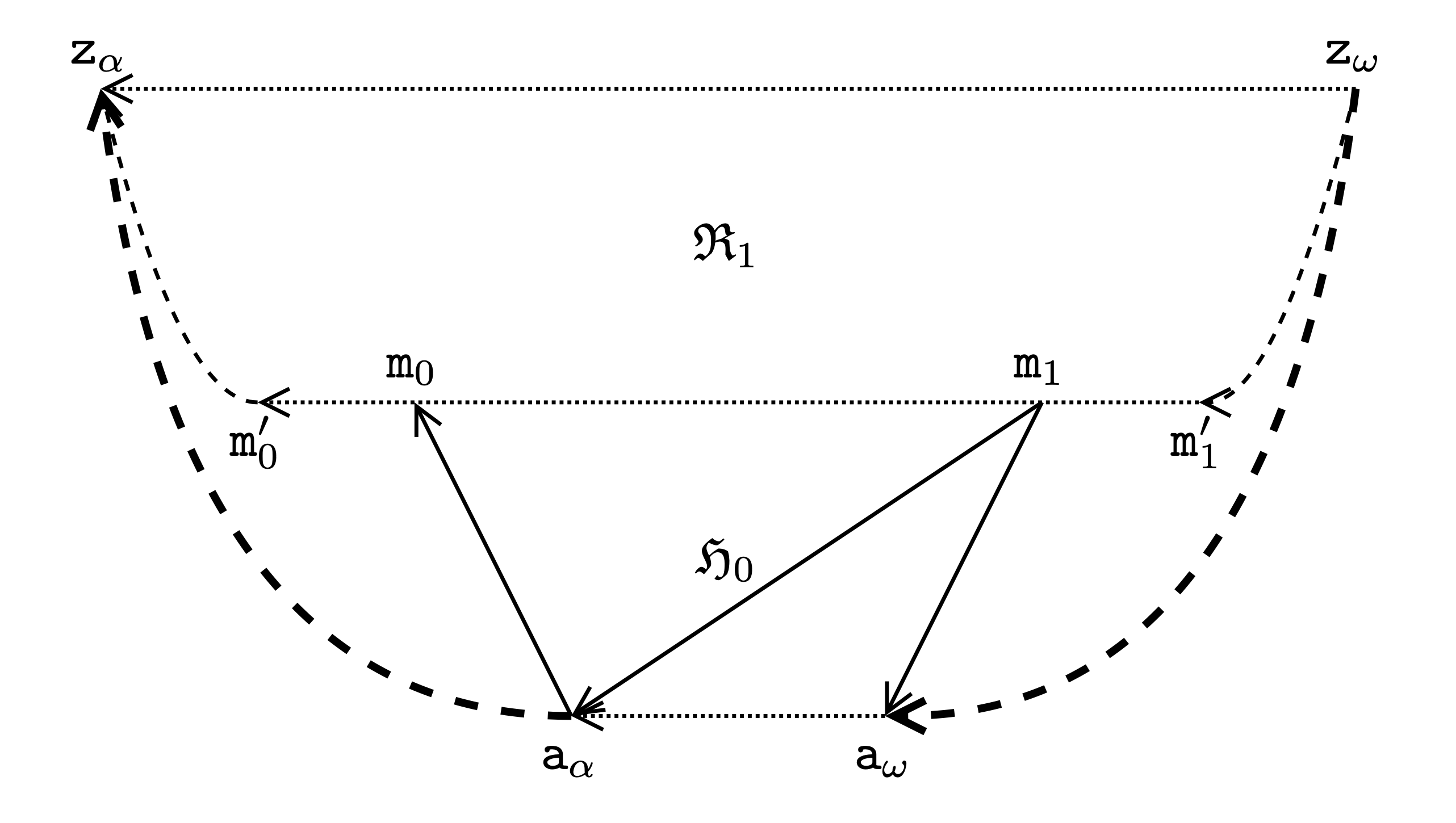}
\end{center}

It's worth noting that the non-duplication of corkline twists guarantees\footnote{Except in the degenerate case of a line tethered to itself, and a rig consisting of those self-tethering hitches spliced together.} 
that twists in a rig's corkline
are not duplicated in other lines in the rig, 
for rigs obtained starting from hitches and applications of the splicing and lashing operations.

This lashing operation allows the construction of a rig whose leadline is not directly tethered to its corkline. 
To account for this we introduce the notion of a tetherline and its height.

\begin{defn}
The \emph{height} of a rig $\mathfrak{R}$ is the least $n$
so that $({^*\tt t})^n({\tt a}_\alpha) \in {\tt C}(\mathfrak{R})$.
\end{defn}

\begin{defn}
The \emph{tetherline} of a rig $\mathfrak{R}$
is the collection of twists connecting ${\tt a}_\alpha$ to the twist
$({^*\tt t})^n({\tt a}_\alpha) \in {\tt C}(\mathfrak{R})$.
\end{defn}

In other words, the tetherline of a rig is the collection of twists which contains:
the first twist in the rig's leadline, and
for each twist ${\tt x}$ in the tetherline, which is not also in the rig's corkline,
that twist's tether ${\tt t}({\tt x})$, that twist's fast tether ${^*\tt t}({\tt x})$,
and all intervening twists {\tt y} with
${\tt t}({\tt x}) \preccurlyeq {\tt y} \preccurlyeq {^*\tt t}({\tt x})$. 

Note that the height of a rig is also the height of its tetherline.

~~~

The splicing and lashing operations
can be used to construct rigs that
have an arbitrary (leadline) length 
and an arbitrary (tetherline) height.
Since the splicing operation only requires 
that the corklines and leadlines of the rigs being spliced be aligned,
without placing any such requirements on intervening lines,
these rigs can be quite flexible.

In particular, the leadline can change the intermediate line 
it is fastened to partway through a rig.
As long as those intermediate lines 
are ultimately rigged up to the corkline,
such moves maintain the corkline's support 
for that leadline.

This is a central feature of TODA files, 
where rigging provides integrity-at-a-distance
while allowing dynamic choice of custodianship:
the source of truth of a file can move from 
server to laptop to phone to a different server,
preserving the integrity of the original server regardless of location.

\section{$G_\uparrow$ is supportive}

We are now ready to introduce $G_\uparrow$, the simplest guild that is closed under the splicing and lashing operations.

\begin{defn}
$G_\uparrow$ is the guild consisting of all rigs that are either
\begin{itemize}
    \item a single half-hitch
    \item a length 1 rig from $G_\uparrow$
    spliced to another rig from $G_\uparrow$
    \item a single half-hitch lashed up to another rig from $G_\uparrow$
\end{itemize}
\end{defn}

Before we show that $G_\uparrow$ is supportive
we need a small result to constrain rig height.

\begin{propn}
\label{tether-prec}
If for a twist {\tt x} and some $n > 0$, $({^*\tt t})^n({\tt x}) \aligned {\tt x}$,
then $({^*\tt t})^n({\tt x}) \prec {\tt x}$.
\end{propn}

\begin{proof}
A twist {\tt x} includes its previous and tether hashes ${\tt t}({\tt x})$ and ${\tt p}({\tt x})$,
so by hash inclusion we have that ${\tt t}({\tt x}) \tprec {\tt x}$ and ${\tt p}({\tt x}) \tprec {\tt x}$.

Because the fast tether function ${^*\tt t}$ is just a composition of {\tt p} and {\tt t},
it follows that $({^*\tt t})^n({\tt x}) \tprec {\tt x}$.

We have already assumed that $({^*\tt t})^n({\tt x}) \aligned {\tt x}$,
so $({^*\tt t})^n({\tt x}) \tprec {\tt x}$ 
implies $({^*\tt t})^n({\tt x}) \prec {\tt x}$,
which is our desired result.
\end{proof}

We are now ready to introduce our main result.

\begin{thm}
\label{main-up-support}
$G_\uparrow$ is supportive.
\end{thm}

\begin{proof}
Once again, we take non-disjoint rigs $\mathfrak{R}$ and $\mathfrak{R}'$ in $G_\uparrow$,
and seek to prove that ${\tt A} \aligned {\tt A}'$.

We also take $l$ and $l'$ to be the (leadline) lengths
and $h$ and $h'$ to be the heights
of $\mathfrak{R}$ and $\mathfrak{R}'$ respectively,
and proceed by structural induction.

~~~

\textbf{Base case ($l = l' = h = h' = 1$):}

In this case, both $\mathfrak{R}$ and $\mathfrak{R'}$ consist of single half-hitches, putting them in $G_H$,
so by Theorem \ref{half-hitch_sound} it follows that ${\tt A} \aligned {\tt A}'$.

~~~

\textbf{Inductive case ($l = l' = 1$, and either $h > 1$ or $h' > 1$):} 

Assuming without loss of generality that ${\tt a}_\alpha \preccurlyeq {\tt a}'_\alpha$,
then having $l = 1$ and non-disjoint rigs 
gives that we either have that ${\tt a}'_\alpha = {\tt a}_\omega$
or that ${\tt a}'_\alpha = {\tt a}_\alpha$.

If ${\tt a}'_\alpha = {\tt a}_\omega$ we simply 
prepend {\tt A} to ${\tt A}'$ to get an enveloping line,
proving that ${\tt A} \aligned {\tt A}'$.

This leaves the case where ${\tt a}'_\alpha = {\tt a}_\alpha$.

We begin by assuming that $h > h'$ and deriving a contradiction.
The definition of $h$ provides that:
\begin{itemize}
    \item $({^*\tt t})^h({\tt a}_\alpha) = {\tt z}_\alpha$ and
    \item $({^*\tt t})^{h'}({\tt a}_\alpha) = {\tt z}'_\alpha$.
\end{itemize}

Thus it follows that $({^*\tt t})^{h-h'}({\tt z}'_\alpha) \aligned {\tt z}_\alpha$.

Our preconditions provide that ${\tt z}_\alpha \aligned {\tt z}'_\alpha$,
and the construction of lashing disallows ${\tt z}_\alpha \preccurlyeq {\tt z}'_\alpha$,
because that would place ${\tt z}'_\alpha$ both in ${\tt C}(\mathfrak{R})$ and also elsewhere in the rig,
so it must be the case that ${\tt z}'_\alpha \prec {\tt z}_\alpha$.

However Proposition \ref{tether-prec}
gives that because ${\tt z}_\alpha = ({^*\tt t})^{h-h'}({\tt z}'_\alpha) \aligned {\tt z}'_\alpha$,
it must be the case that ${\tt z}_\alpha \prec {\tt z}'_\alpha$.

This is a contradiction, so it cannot be the case that $h > h'$.
A similar argument excludes the case where $h' > h$, so the heights must be equal.

In the $h = h'$ case, we know that a single half-hitch has $h=1$, but $h > 1$
so $\mathfrak{R}$ is not a single half-hitch.
Further, a rig constructed by splicing together two rigs has a length
that is the sum of the lengths of the spliced rigs;
since every rig has a length of at least one,
a spliced rig must have a length of at least two.
Thus with $l = 1$, $\mathfrak{R}$ cannot be constructed by splicing two rigs together.
So we must be in the remaining case, where $\mathfrak{R}$ consists of
a half-hitch $H$ lashed up to a rig $\mathfrak{S}$, with a height of $h-1$.

Similarly, $\mathfrak{R}'$ must consist of a half-hitch $H'$ lashed up to a rig $\mathfrak{S}'$, with a height of $h'-1$. 

By the definition of lashing, we have that the leadline {\tt B} of $\mathfrak{S}$
(and similarly the leadline ${\tt B}'$ of $\mathfrak{S}'$) must begin with ${^*\tt t}({\tt a}_0)$,
and with $\mathfrak{S}$ and $\mathfrak{S}'$ having the same corklines
as $\mathfrak{R}$ and $\mathfrak{R}'$ respectively,
we know them to be non-disjoint.
Thus our induction hypothesis provides that 
$\mathfrak{S}$ and $\mathfrak{S}'$ cannot be misaligned,
so it follows that ${\tt B} \aligned {\tt B}'$.

This alignment and the existence of a common twist between {\tt A} and ${\tt A}'$
gives that $\mathfrak{H}$ and $\mathfrak{H}'$ are non-disjoint,
and since they both belong to the supportive guild $G_H$,
they must therefore be aligned.
Thus we have alignment between their leadlines, i.e. ${\tt A} \aligned {\tt A}'$.

~~~

\textbf{Inductive case ($l,l' \geq 1$):}

Like with $h$ in the previous case, we assume without loss of generality that $l > 1$.
Because a single half-hitch has length $l = 1$, and lashings also yield rigs with length $l = 1$,
it must be the case that $\mathfrak{R}$ is obtained by splicing.

The induction hypothesis provides that the spliced rigs satisfy the conditions of Theorem \ref{splicing-sound},
which provides that ${\tt A} \aligned {\tt A}'$.

Thus by structural induction, we have that $G_\uparrow$ is supportive.
\end{proof}

Theorem \ref{main-up-support} shows that the guild $G_\uparrow$ is supportive.
Proposition \ref{unique-canonical-succession} shows that
a supportive guild guarantees unique canonical succession in a supported line.

~~~

\noindent
To put this in the vernacular:

\begin{cor}
Rigs in $G_\uparrow$ prevent double spend. 
\end{cor}

\newpage

\section{Glossary}

We introduce and use a number of data structures, as well as functions on those structures.
To maintain legibility we adhere to the following conventions
\begin{itemize}
    \item twists are represented by lower-case monospace; e.g. {\tt t} is a twist
    \item lines are represented by uppercase monospace; e.g. {\tt A} is a line
    \item lines can also be represented as square brackets containing a list of twists (in chronological order);
        e.g. $[{\tt a}_0,{\tt a}_1,{\tt a}_2,{\tt a}_3]$ 
    \item hitches are represented by upper-case italics; e.g. $H$ is a hitch
    \item guilds are also represented by upper-case italics; e.g. $G$ is a guild
    \item rigs are represented by uppercase fraktur; e.g. $\mathfrak{R}$ is a rig
    \item other values are represented by lowercase italics (e.g. simple numbers, hash outputs, arbitrary inputs, tries); e.g. $n$
\end{itemize}

Further, the symbols for functions follow the same convention that is used for their images;
e.g ${\tt f}(\mathfrak{R})$ is a twist-valued function on a rig.

Common functions are
\begin{itemize}
    \item ${\tt t}({\tt x})$; the tether of {\tt x}
    \item ${^*\tt t}({\tt x})$; the fast tether of {\tt x}
    \item ${\tt p}({\tt x})$; the previous of {\tt x}
    \item ${^*\tt p}({\tt x})$; the previous of {\tt x}
    \item $r({\tt x})$; the rigging trie of {\tt x}
    \item ${\tt f}(H)$; the fastener of $H$
    \item ${\tt h}(H)$; the hoist of $H$
    \item ${\tt l}(H)$; the lead of $H$
    \item ${\tt m}(H)$; the meet of $H$
    \item ${\tt n}(H)$; the post of $H$
    \item $h(x)$; the hash of $x$
\end{itemize}

Common relations and operations are:  
\begin{itemize}
    \item ${\tt x} \prec {\tt y}$ is read ``x is a predecessor of y" or ``x precedes y", and means twist ${\tt x}$ is on the line defined by twist ${\tt y}$;
    \item ${\tt x} \aligned {\tt y}$ is read ``x is aligned with y" or ``x and y are aligned" and means that
    ${\tt x} \preccurlyeq {\tt y}$ or ${\tt y} \preccurlyeq {\tt x}$;
    \item ${\tt A} \aligned {\tt B}$ is read ``A is aligned with B" or ``A and B aligned", and means line ${\tt A}$ and line ${\tt B}$ are both contained within a line ${\tt C}$;
    \item ${\mathfrak{R} \aligned \mathfrak{R}'}$ is read ``R and R prime are aligned", and means
    ${\tt L}(\mathfrak{R}) \aligned {\tt L}(\mathfrak{R}')$;
    \item $x \tprec y$ is read ``x causally precedes y", and means that $y$ includes $x$ through one-way functions, for instance by containing the cryptographic hash of $x$.
\end{itemize}

When comparing two rigs $\mathfrak{R}$ and $\mathfrak{R}'$,
{\tt A}/{\tt a} and {\tt Z}/{\tt z} are used for their leadlines and corklines
\begin{itemize}
    \item ${\tt A} = {\tt L}(\mathfrak{R}) = [{\tt a}_\alpha,...,{\tt a}_\omega]$
    \item ${\tt Z} = {\tt C}(\mathfrak{R}) = [{\tt z}_\alpha,...,{\tt z}_\omega]$
    \item ${\tt A}' = {\tt L}(\mathfrak{R'}) = [{\tt a}'_\alpha,...,{\tt a}'_\omega]$
    \item ${\tt Z}' = {\tt C}(\mathfrak{R'}) = [{\tt z}'_\alpha,...,{\tt z}'_\omega]$
\end{itemize}
Additionally, when $\mathfrak{R}$ and $\mathfrak{R}'$ are non-disjoint,
{\tt z} is the most recent of ${\tt z}_\omega$ and ${\tt z}'_\omega$
and {\tt a} is the common twist on their leadlines.

\section{Acknowledgements}

We wish to particularly thank Toufi Saliba for cofounding the TODA vision, Hassan Khan for relentlessly championing it, and Adam Gravitis for engineering insight during the long gestation of this work. Thank you to the following reviewers for their valuable feedback: Cory Sulpizi, Mike Everson, Alexander Fertman, and Robert Moir.

The authors are grateful to TODAQ for their generous financial support of this work.

\bibliographystyle{amsplain}
\bibliography{rigging}

\end{document}